\documentclass{amsart}
\usepackage[latin9]{inputenc}
\usepackage{geometry}
\usepackage{verbatim}
\usepackage{amstext}
\usepackage{amsthm}
\usepackage{amssymb}
\usepackage{graphicx}
\usepackage[hidelinks=true]{hyperref}

\makeatletter

\newcommand{\R}{ {\mathbb{R}} }
\newcommand{\E}{ {\mathbb{E}} }
\renewcommand{\P}{ {\mathbb{P}} }

\newtheorem{theorem}{Theorem}\theoremstyle{plain}

\renewcommand{\and}{\quad\textrm{ and }\quad}

\renewcommand{\P}{\mathbb{P}}

\newcommand{\N}{\mathbb{N}}

\usepackage{xcolor}

\newtheorem{lemma}{Lemma}\newtheorem{proposition}{Proposition}\numberwithin{equation}{section}

\title{On the martingale property 
 in the rough Bergomi model}

\author{Paul Gassiat}
\address{
Universit\'e Paris-Dauphine, PSL University, UMR 7534, CNRS, CEREMADE, 75016 Paris, France}
\email{gassiat@ceremade.dauphine.fr}

\begin{document}

\begin{abstract}
We consider a class of fractional stochastic volatility models (including the so-called rough Bergomi model), where the volatility is a superlinear function of a fractional Gaussian process. We show that the stock price  is a true martingale if and only if the correlation $\rho$ between the driving Brownian motions of the stock and the volatility is nonpositive. We also show that for each $\rho<0$ and $m> \frac{1}{{1-\rho^2}}$, the $m$-th moment of the stock price is infinite at each positive time.

\end{abstract}

\thanks{The author is grateful to P. K. Friz for related discussions. This work is partially supported by the ANR via the project ANR-16-CE40- 0020-01. }

\maketitle


\section{Introduction and main results} \label{sec:Intro}

We are interested in fractional stochastic volatility models where the dynamics of (discounted) stock price under a risk-neutral measure take the form 
\begin{equation}
dS_{t}/S_{t}=\sigma(t,Y_t) dW_t%
  \label{equ:S}
\end{equation}%
\begin{equation}
Y_{t} =\int_{0}^{t} K(t,s) dZ_{s}
\label{eq:Y}
\end{equation}
where $Z_t =\rho W_{t}+\bar{\rho}\bar{W}_t$, $W,\bar{W}$ are two independent Brownian motions on a filtered probability space $\left(\Omega,(\mathcal{F}_t)_{t\geq0}, \mathbb{P}\right)$, and $\rho ^{2}+\bar{\rho}^{2}=1$.

A specific example we have in mind is the so-called Rough Bergomi model introduced in \cite{BFG16}. In that model $Y$ is a Riemann-Liouville fractional Brownian motion of Hurst parameter $H \in (0,1)$, i.e. 
$$K(t,s) = C_H (t-s)^{H - \frac1 2}$$
and the volatility function takes the form
$$\sigma(t,y) = \zeta(t) \exp\left( \eta y \right)$$
where $\eta >0$ and $\zeta$ is a continuous function of $t$. The rough Bergomi model with $H \in (0,\frac{1}{2})$ is part of a larger class of fractional stochastic volatility models (so-called "rough volatility models") which has been recently observed to reproduce several features of historical \cite{GJR18} and pricing \cite{ALV07,F11,BFG16} data, and has been the subject of intense  recent academic activity\footnote{See for instance the website \url{https://sites.google.com/site/roughvol/home} for an up to date listing of the relevant literature.}.

The first question we consider in this note is whether the price process $S$, which is obviously a local martingale (and a supermartingale) is a true martingale. The true martingale property is very important in practice, since using a strict local martingale measure for pricing has some obvious drawbacks. For instance : if $S$ is a strict local martingale then $\E[S_T]<S_0$ for some $T>0$, so that already the price given by the model for holding one unit of stock until time $T$ does not coincide with market data (this suggests that the asset price is greater than its actual "fundamental" value and for this reason strict local martingale models have been used in the modelisation of bubbles, see \cite{Pro13} and references therein).

Note that in the rough Bergomi model, due to the superlinear growth of $\sigma$, Novikov's criterion for martingality is never satisfied. Nevertheless we show that if the correlation is nonpositive (which is typically the case in actual applications), the price process is indeed a true martingale. Actually our result does not rely on the specific form of $\sigma$ in the rough Bergomi model, and only requires a (rather weak) assumption on $K$ and $\sigma$. 
We also show the converse implication in the case of a Riemann-Liouville type kernel, under a more specific assumption  of superlinear growth of $\sigma$. 

\begin{theorem} \label{thm:mart}
(1) Assume that  the kernel $K$ is such that \eqref{equ:sigma} defines a Gaussian process with continuous sample paths, $\sigma : [0,\infty) \times \R \to \R_+$ is continuous and bounded on $[0,T]\times (-\infty, a]$ for each $T$, $a$ $>$ $0$. Then if $\rho \leq 0$, $(S_t)_{t \geq 0}$ defined by \eqref{equ:S}-\eqref{eq:Y} is a true martingale. 

(2) Assume in addition that there exists $T_0 >0$ such that for some $\alpha > \frac{1}{2}$, 
$$\forall 0\leq s\leq t \leq T_0, \;\; K(t,s) { = K(t-s,0)} \geq \alpha (t-s)^{\alpha-1}$$
and $\sigma \geq \sigma_0$ on $[0,T_0] \times \R$, where $\sigma_0 : [0,T_0] \times \R \to \R_+$ is  continuous,  nondecreasing in $x$, locally Lipschitz in $x$ (uniformly in $t \in [0,T_0]$)  and such that for some $A>0$,
 \begin{equation}
\label{eq:asnG}
\int_A^{+\infty} \left(\frac{w}{\inf_{t \in [0,T_0]} \sigma_0(t,w)}\right)^{\frac{1}{\alpha}} \frac{dw}{w} < + \infty,
\end{equation}
Then if $\rho>0$, for each $t>0$, one has $\E[S_t] < S_0$.\end{theorem}

\vspace{3mm}

The second result of this note deals with the moments of the stock price. We show that, under a similar assumption (satisfied by the rough Bergomi model), for each value of $\rho  \in (-1,0]$, some of the higher order moments are infinite.

\begin{theorem} \label{thm:mom}

Assume that there exists $T_0>0$  s.t. for some $\alpha > \frac 1 2$, $K(s,t) = \alpha (t-s)^{\alpha-1}$ for all $0\leq s \leq t \leq T_0$ and $\sigma = \sigma_0
$ on $[0,T_0] \times \R$ with $\sigma_0$ as in Theorem \ref{thm:mart} (2). 
Then if $\rho \leq 0$, $m> 1$ are such that $\rho^2 <\frac{m-1}{m}$, it holds that for all $t>0$, $\E[S_t^m]=+\infty$.
\end{theorem}

The finiteness of moments is important for instance in Monte Carlo simulation (to know that the Monte Carlo error is ruled by CLT estimates, finite variance is  needed) and in asymptotic formulae (to go from stock price large deviations to call price asymptotics, see for instance \cite[section 4.2]{FGP18+}). It would therefore be very useful to obtain a converse (positive) result.

\vspace{1mm}
We remark that in the Brownian case ($K \equiv 1$), both of the above results are well known, cf. \cite{Sin98,Jou04,LM07}, and in that case the condition $\rho^2 >\frac{m-1}{m}$ is also a sufficient condition for the moments to be finite. In the case of the rough Heston model, \cite{GGP18+} have recently obtained some results on moments of the stock price (which are similar in spirit to those for classical Heston, and therefore quite different from ours). {Finally, we would like to point out that a more general version of Theorem \ref{thm:mom} has been obtained independently and by a different method in \cite{Gul19+}.}

\vspace{1mm}
The remainder of this note is devoted to the proofs of Theorems \ref{thm:mart} and \ref{thm:mom}, which we now outline. The proof of Theorem \ref{thm:mart} follows the classical argument (found already in the aforementioned \cite{Sin98,Jou04,LM07}, see also \cite{BE09,Ruf15} for additional references) relating the martingale property of stochastic exponentials with explosions of a SDE (in our case, this will be a Volterra SDE). The martingale property (1) is then essentially immediate, while the proof of (2) follows from the fact that the Volterra SDE may blow up in arbitrarily short time with positive probability.  

The proof of Theorem \ref{thm:mom} relies on the Bou\'e-Dupuis formula, which expresses the expectation of exponentials of Brownian functionals as values of (here : Volterra) stochastic control problems. We then show that for the considered values of the parameters, we may choose a feedback control such that, as in the previous proof, the process (and the value) blow up in arbitrarily small time. This proof is new even in the classical (Markovian) case.

\section{Proofs}

\subsection{Preliminaries}

\subsubsection{Volterra integral equations}

In this subsection, we fix 
$$K_\alpha(r) = \alpha r^{\alpha-1} \mbox{ for some }\alpha >0,$$
$$z : [0,\infty) \to \R \mbox{ continuous}$$
$$b:[0,\infty)\times \R \to \R_+ \mbox{ continuous }$$
and consider the Volterra equation
\begin{equation} \label{eq:VoltDem}
y(t) = z(t) + \int_0^t K_\alpha(t-s) b(s,y(s)) ds, \;\;\ t \geq 0,
\end{equation}
of unknown $y$.

We will use the following results.

\begin{proposition} \label{prop:VoltLip}
Assume that $b$ is Lipschitz continuous in $x$, uniformly in $t \in [0,T]$, for each $T>0$. Then \eqref{eq:VoltDem} admits a unique continuous solution $y$ on $[0,\infty)$.
\end{proposition}

\begin{proof}
Uniqueness is easy to check directly, and existence follows from \cite[Theorem 12.2.8]{GLS}.
\end{proof}

\begin{proposition} \label{prop:VoltMonotone}
Assume that $b$ is nondecreasing in $x$ and locally Lipschitz in $x$ (uniformly in $t\in [0,T]$ for each $T>0$).

Then :
\begin{enumerate}
\item there exists a unique pair $(y,T_\infty)$ such that $y$ is a continuous solution of \eqref{eq:VoltDem} on $[0,T_\infty)$, and $\lim_{t \to T_\infty} y(t) = +\infty$.
\item Let $T< T_\infty$, $u$ $:$ $[0,T] \to \R$ continuous such that 
$$u(t) \leq \mbox{(resp. } \geq \mbox{)  }  z(t) + \int_0^t K_\alpha(t-s) b(s,u(s)) ds ,\; \forall 0 \leq t \leq T.$$
Then $u \leq y$ (resp. $u \geq y$) on $[0,T]$.
\item Assume that $t\mapsto z(t)$ and $t\mapsto b(t,x)$ are nondecreasing (for each $x \in \R$). Then so is $t \mapsto y(t)$.
\end{enumerate}
\end{proposition}

\begin{proof}
cf. \cite[Theorems 13.5.1 and 13.4.7]{GLS} and \cite[Theorem 2.6]{BY12}.
\end{proof}

We will also use the following lemma which gives an explicit upper bound on blow-up time for solutions to \eqref{eq:VoltDem}.

\begin{lemma} \label{lem:expl}
In the setting of Proposition \ref{prop:VoltMonotone}, assume that $t\mapsto z(t)$ is nondecreasing. Then with $T_{\infty}$ as in (1), it holds that for each $T>0$
\begin{equation}
T_\infty \wedge T \leq \inf_{x \geq 0} \left( h(x) + \frac{1}{\alpha}\int_x^\infty \left( \frac{w}{ \inf_{t \in [0,T]} b(t,w)}\right)^{\frac{1}{\alpha}} \frac{dw}{w} \right), 
\end{equation} 
where $h(x) = \sup \{t : z(t) \leq x \}$.
\end{lemma}

\begin{proof}
We follow arguments from \cite{BY12}. By Proposition \ref{prop:VoltMonotone} (2), it suffices to consider the solution to \eqref{eq:VoltDem} when $b$ is replaced by $b_0:= \inf_{t \in [0,T]} b(t,\cdot)$.

We fix $x \geq 0$ such that $h(x)$ and $\int_x^\infty \left( \frac{w}{b_0(w)}\right)^{\frac{1}{\alpha}} \frac{dw}{w} $ are finite,  $R>1$, and for each $n \geq 0$ we let 
$$T_n = \sup\{ t : y(t) \leq x R^n \} \in (0,+\infty].$$
Note that $T_0 \leq h(x)$. We then have for $n \geq 1$, for each $t>T_{n-1}$
\begin{align*}
 y(T_n \wedge t) = x R^n &= z(T_n\wedge t) +  \int_0^{T_n\wedge t} \alpha ((T_{n}\wedge t)-s)^{\alpha - 1} b_0(y(s)) ds \\
 &\geq z(T_n\wedge t) + \int_{T_{n-1}}^{T_n\wedge t} \alpha ((T_{n}\wedge t)-s)^{\alpha - 1} b_0(y(s)) ds  \\
 &\geq xR^{n-1} +  \left(T_{n}\wedge t-T_{n-1}\right)^\alpha b_0( x R^{n-1})
\end{align*}
where we have used the monotonicity of $y$ (Proposition \ref{prop:VoltMonotone} (3)). This implies that if $T_{n-1}$ is finite, so is $T_n$, with
$$T_n \leq T_{n-1} + \left( \frac{x (R^n-R^{n-1})}{b(x R^{n-1})}\right)^{\frac{1}{\alpha}}.$$
Hence 
\begin{align*}
T_\infty &= T_0 + \sum_{n \geq 1} (T_n - T_{n-1}) \\
&\leq h(x) + \sum_{n \geq 1} \left( \frac{x R^n - xR^{n-1}}{b_0(x R^{n-1})}\right)^{\frac{1}{\alpha}} = h(x) + \sum_{n\geq 1} \frac{1}{\alpha} \int_{x R^{n-1}}^{x R^n} \frac{w^{\frac{1}{\alpha}-1}}{b_0(xR^{n-1})} dw.\end{align*}
We obtain the result by letting $R \downarrow 1$.
\end{proof}

\subsubsection{Stochastic convolutions}

We consider the following stochastic convolution
\begin{equation}
Y_{t} =\int_{0}^{t} K(t,s) dZ_{s}
\label{equ:sigma}
\end{equation} 
(recall that $(Z_t)_{t \geq 0}$ is a $\P$-Brownian motion).

We recall the well-known condition for $Y$ to be continuous (cf. e.g. \cite[Theorem 2.1]{GRR70}).

\begin{proposition} \label{prop:ContStoch}
Assume that 
$$\forall t >0, C_K(t,t):= \int_0^t K(t,s)^2 ds < \infty,$$
{and for $t,t'\geq 0$  let $C_K(t,t') = \int_0^{t \wedge t'} K(t,s)K(t',s) ds.$}
Assume that for each $T \geq 0$, letting
$$\theta_T(h) := \sup_{0\leq t, t' \leq T, |t-t'|\leq h} \left\{ C_K(t,t) + C_K(t',t') - 2 C_K(t,t')\right\}^{1/2}$$
it holds that { $\theta_T(0^+)=0$ and}
$$\int_{0^+} \sqrt{\ln(1/u)} d\theta_T(u) < \infty.$$
Then $Y$ admits a version with continuous sample paths.
\end{proposition}

Note that the assumption above is satisfied for $K(t,s) = \alpha (t-s)^{\alpha-1}$,  any $\alpha > - \frac{1}{2}$.

We will also need a result on the support of the law of $Y$ { when $K$ is translation invariant.

\begin{proposition} \label{prop:Support}In addition to the assumptions of Proposition \ref{prop:ContStoch}, assume that $K(t,s)= \hat{K}(t-s)$ for all $0\leq s \leq t$, where $\hat{K}$ is a function s.t. $\int_0^\varepsilon |\hat{K}| >0$ for all $\varepsilon >0$. Then for each $T\geq 0$, the law of $Y$ has full support in $C_0^T :=\{y \in C([0,T],\R); y(0)=0\}$ (equipped with the topology of uniform convergence).
\end{proposition}

\begin{proof}
Note that the law of $Y$ is a Gaussian measure on $C_0^T$, with Cameron-Martin space
$$\mathcal{H}_K = \left\{ y^f, \;\; f \in L^2([0,T])\right\} \subset C_0^T$$
 where for $f \in L^2([0,T]$ we define
 $$y^f : t \mapsto \int_0^t \hat{K}(t-s) f(s) ds.$$
 By a classical application of the Cameron-Martin theorem (see e.g. \cite[Theorem 3.6.1]{bogachev}), the support of the law of $Y$ in $C_0^T$ is the closure of $\mathcal{H}_K$ in $C_0^T$. We then conclude with \cite[Lemma 2.1]{C08}.
%
\end{proof}
}

\subsection{Proof of Theorem \ref{thm:mart}}

Since $(S_t)$ is a nonnegative local martingale (hence supermartingale), it will be a martingale on $[0,T]$ if and only if $\E[S_T] = S_0$.

Letting $\tau_n = \inf\{ t>0, \;Y_t = n \}$, then since $\sigma$ is bounded on $[0,T]\times (-\infty,n]$ it holds that
$$S_0 = \E\left[ S_{T \wedge \tau_n} \right] = \E\left[ S_T  1_{\{T < \tau_n\}}\right] + \E\left[ S_{\tau_n}  1_{\{\tau_n \leq T\}}\right].$$
The first term converges to $\E[S_T]$ when $n$ goes to infinity, so that
\begin{equation} \label{eq:lim}
S_0 - \E\left[S_T\right] = \lim_{n \to \infty}  \E\left[ S_{\tau_n}  1_{\{\tau_n \leq T\}}\right].
\end{equation}
On the other hand, we can apply Girsanov's theorem to write
$$ \E\left[ S_{\tau_n}  1_{\{\tau_n \leq T\}}\right] = S_0 \hat{\P}_n\left(\tau_n \leq T \right)$$
where 
 $\hat{\P}_n$ is such that 
$$\hat{W}^{(n)}_t = W_t - \int_0^{t \wedge \tau_n} \sigma(s,Y_s) ds$$
is a Brownian motion under $\hat{\P}_n$.  
Note that for $t \leq  \tau_n$ one has
\begin{align*}
Y_t &= \int_{0}^{t}K(t,s) \left( d\hat{Z}^{(n)}_{s} + \rho \sigma(s,Y_s) ds \right) \\
&= \hat{Y}_t + \int_0^t K(t,s) \rho \sigma(s,Y_s) ds
\end{align*}
where $\hat{Z}^{(n)}$ is a $\hat{\P}_n$-Brownian motion, and
\[
\hat{Y}_t := \int_{0}^{t}K(t,s)d\hat{Z}^{(n)}_{s}.
\]

We first treat the case $\rho \leq 0$. Since $Y_t \leq \hat{Y}_t$ (for $t\leq \tau_n$) one then has $ \tau_n \geq \tau_n^0 := \inf\{ t>0, \; \hat{Y}_t = n \}$. In addition, since $\hat{Z}^{(n)}$ is a $\hat{\P}_n$-Brownian motion, one has
$$\lim_{n \to \infty} \hat{\P}_n(\tau^0_n \leq T)  = \lim_{ n \to \infty}  \P (\sup_{t \in [0,T]} Y_t \geq n) = 0,$$
and it follows that
$$S_0 - \E[S_T] = S_0  \lim_{n \to \infty} \hat{\P}_n(\tau_n \leq T) = 0$$
i.e. $S$ is a martingale.

%

We now treat the case $\rho > 0$. We then have for $t < \tau_n$
\begin{equation} \label{eq:lambdax}
Y_t \geq \hat{Y}_t +  \int_{0}^{t} \alpha \left( t-s\right)^{\alpha-1} \rho \sigma_{0}(s,Y_s) ds.
\end{equation}
In particular, by Proposition \ref{prop:VoltMonotone} (2) and the fact that $\hat{Z}^{(n)}$ is a Brownian motion under $\hat{P}_n$, one has 
$$\lim_{n \to \infty} \hat{\P}_n(\tau_n \leq  T) \geq \P(T_\infty < T)$$
where $T_\infty$ is the explosion time of the solution $\tilde{Y}$ to 
$$\tilde{Y}_t = Y_t +  \int_{0}^{t} \alpha \left( t-s\right)^{\alpha-1} \rho \sigma_{0}(s,Y_s) ds$$
(which exists and is unique $\P$-a.s. by Proposition \ref{prop:VoltMonotone}).

Let $x, \lambda >0$ be chosen such that 
$$\frac{x+1}{\lambda} + \int_{x} ^\infty \left(\frac{w}{\rho \cdot \inf_{t \in [0,T]} \sigma_0(t,w)}\right)^{\frac{1}{\alpha}} \frac{dw}{w}  < T.$$
Let $z_\lambda(t) = \lambda t -1$ and $y_\lambda$ be the solution to \eqref{eq:VoltDem} with $z=z_\lambda$ and $b(t,\cdot)=\sigma_0$. By Lemma \ref{lem:expl}, $y_\lambda$ blows up on $[0,T]$.

By Proposition \ref{prop:Support}, the event $\left\{ Y \geq z_\lambda \mbox{ on }[0,T]\right\}$ has positive probability under ${\P}$. But on this event, one has $\tilde{Y} \geq y_\lambda$ on $[0,T]$, and $T_\infty < T$. This proves that $\E[S_T] < S_0$.

\subsection{Proof of Theorem \ref{thm:mom}}

We again apply a Girsanov transformation : one has
\begin{align*}
\E\left[S_T^m\right] &=  S_0^m \;\E \left[ \exp\left(\int_0^T m \sigma(s,Y_s) dW_s-\int_0^T \frac{m}{2} \sigma^2(s,Y_s) ds\right)\right]\\
&= S_0^m \;  \hat{\E} \left[ \exp\left(\int_0^T \frac{m^2-m}{2} \sigma^2(s,Y_s) ds\right)\right],
\end{align*}
with 
$$\frac{d\hat{\P}}{d\P} = \exp\left(\int_0^T m \sigma(s,Y_s) dW_s-\int_0^T \frac{m^2}{2} \sigma^2(s,Y_s) ds\right)$$
(this defines a probability measure by Theorem \ref{thm:mart} (1) since $\rho\leq 0$), and we have that 
\[
Y_t = Y_0 + \int_0^t K_\alpha(t-s) (d\hat{W}_s + \rho m \sigma(s,Y_s) ds )\]
for a $\hat{\P}$-Brownian motion $\hat{W}$. Letting $Y^0_t = Y_0 + \int_0^t K_\alpha(t-s) d \hat{W}_s$, this is rewritten as
\begin{equation} \label{eq:YY0}
Y_t = Y_t^0 + \int_0^t K_\alpha(t-s) \rho m \sigma(s,Y_s) ds.
\end{equation}
Note that since $Y^0$ is $\hat{\P}$-a.s. continuous, combining Proposition \ref{prop:VoltLip} with the a priori bounds
$$ Y_t^0 + \rho m \left(\int_0^t K_\alpha\right) \sup_{s \in [0,t]}\sigma(s,\sup_{s\leq t} Y^0_s) \leq  Y_t \leq Y^0_t $$
one can show that \eqref{eq:YY0} admits $\hat{\P}$-a.s. a unique continuous solution.

By the Boué-Dupuis formula \cite[Theorem 5.1]{BD98}, this yields
\[ \ln {\E}\left[S_T^m/S^m_0\right]  = \sup_{(v_t)_{t\geq0} \in \mathcal{V}} \hat{\E} \left[ \int_0^T  \left(\frac{m^2-m}{2} \sigma^2(s,Y_s^v) - \frac{v_s^2}{2} ds\right) \right]
\]
where 
\[ \mathcal{V} = \left\{ (v_t)_{t \geq 0} \mbox{ progressively measurable with }\hat{\E}\left[ \int_0^T v_t^2 dt \right] < +\infty \right\} 
\]
and for $v \in \mathcal{V}$, $Y^v$ is the unique continuous solution to
\[Y_t^v = Y_0 + \int_0^t K_\alpha(t-s) (d\hat{W}_s + (\rho m \sigma(s,Y_s) + v_s) ds ).\]

On the other hand, if $\rho^2 < \frac{m-1}{m}$ one can find $\gamma$ such that
\[ \rho m + \gamma >0 ,\;\;\;m^2 - m - \gamma^2 >0.\]

The idea is then that taking the feedback control $v_s = \gamma \sigma(Y_s)$, using the first inequality, it holds that for each $T>0$ $Y^v$ has positive probability of blowing up before $T$. On the other hand, the second inequality ensures that the gain is $+\infty$ in this case, so that the value (and the moment) is infinite.

We now give a rigorous proof. We fix $A>0$, $n>0$, let $\theta_A = \inf\{t ; Y^0_t \geq A\}$ and define 
\[
v^{n,A}_s = \begin{cases}  
\gamma \sigma(s,Y_s^{n,A}) &\mbox{ if } (\rho m + \gamma) \sigma(s,Y_s^{n,A}) \leq n \mbox{ and } s \leq \theta_A, \\
n - \rho m \sigma(s,Y^{n,A}_s)  & \mbox{ if }  (\rho m + \gamma) \sigma(s,Y_s^{n,A}) > n \mbox{ and }s \leq \theta_A \\
0 & \mbox{ if } s > \theta_A
\end{cases}
\]
where $Y^{n,A}$ is the unique (by Proposition \ref{prop:VoltLip}) solution to
\[
Y_t^{n,A} = Y^0_t+ \int_0^{t \wedge \theta_A}  K_\alpha(t-s)     \left[(\rho m + \gamma) \sigma(s,Y^{n,A}_s) \wedge n\right] ds + \int_{t \wedge \theta_A}^t  K_\alpha(t-s) \rho m  \sigma(s,Y^{n,A}_s) ds
\]
{
Note that for all $t \in [0,T]$ one has $Y_t^{n,A} \leq Y^0_t + n  t^\alpha$  and also
\begin{equation} \label{eq:v}
 0 \leq v_t^{n,A} \leq \gamma \sigma(t,Y_t^{n,A}) 1_{t \leq \theta_A} \leq \gamma \sup_{t \in [0,T]} \sigma(t,A + n t^\alpha),
\end{equation}
so that in particular $v^{n,A} \in \mathcal{V}$.
}We therefore have
\begin{align}
\ln {\E}\left[S_T^m/S^m_0\right] & \geq  \hat{\E} \left[ \int_0^T  \left(\frac{m^2-m}{2} \sigma^2(s,Y_s^{n,A}) - \frac{(v_s^{n,A})^2}{2} ds\right) \right]  \nonumber \\
&\geq \hat{\E} \left[ 1_{\theta_A >T} \int_0^T  \frac{m^2-m - \gamma^2}{2} \sigma^2(s,Y_s^{n,A}) ds \right].\label{eq:An}
\end{align}
{
where we have used the second inequality in \eqref{eq:v}.}
Now as in the proof of Theorem \ref{thm:mart}, we fix $x,\lambda$ such that
$$\frac{x+1}{\lambda} + \int_{x} ^\infty \left(\frac{w}{(\rho m + \gamma) \inf_{t\in[0,T]} \sigma(t,w)}\right)^{\frac{1}{\alpha}} \frac{dw}{w}  < T,$$
let $z_\lambda(t) = \lambda t -1$ and for $n \in \N \cup \{\infty\}$ let $y_\lambda^n$ be the solution to \eqref{eq:VoltDem} with $z=z_\lambda$ and $b(t,\cdot)=(\rho m + \gamma) \sigma_0 \wedge n$. Note that $y_\lambda^\infty$ blows up in time $T_\infty< T$ by Lemma \ref{lem:expl}. By Proposition \ref{prop:VoltMonotone} (2), $y^n_\lambda$ is nondecreasing in $n$, and therefore for each $T_\infty < t < T$, $y^n_\lambda(t) \to_{n \uparrow \infty} +\infty$.

Fix $A = \lambda T +1 $. On the event $\{ z_\lambda \leq Y^0 \leq z_{\lambda} +1\}$, it holds that $\theta_A >T$, and by Proposition \ref{prop:VoltMonotone} (2), $Y^{n,A}_t \geq y_\lambda^n(t) \to +\infty$ on $[T_\infty,T]$. Letting $n \uparrow \infty$ in \eqref{eq:An} we obtain
\begin{align*}
 \ln \E\left[S_T^m/S_0^m\right]& \geq \infty \cdot \hat{\P}\left(z_\lambda < Y^0  < z_\lambda+1 \mbox{ on }[0,T] \right),
\end{align*}
and we can conclude since the above probability is nonzero by Proposition \ref{prop:Support}.

\bibliographystyle{plain}
\bibliography{roughmoments}

\end{document}